\documentclass[12pt]{article}

\newcommand{\remove}[1]{}

\usepackage[cm]{fullpage}%
\usepackage[breaklinks]{hyperref}  
\usepackage{graphicx}
\usepackage{picins}
\usepackage{amstext}
\usepackage{salgorithm}%
\usepackage{euscript}%
\usepackage{mleftright}%
\usepackage{xspace}%
\usepackage{xcolor}%

\usepackage{amsmath}%
\usepackage{amssymb}%
\usepackage[amsmath,thmmarks]{ntheorem}%
\theoremseparator{.}%

\renewcommand{\Re}{\mathbb{R}}%

\definecolor{blue25}{rgb}{0, 0, 11}
\newcommand{\emphic}[2]{%
   \textcolor{blue25}{%
      \textbf{\emph{#1}}}%
   \index{#2}}

\newcommand{\emphi}[1]{\emphic{#1}{#1}}

   \newtheorem{theorem}{Theorem}[section] 
   
   \newtheorem{lemma}[theorem]{Lemma}

   \newcommand{\myqedsymbol}{\rule{2mm}{2mm}}

   \newcommand{\pbrcx}[1]{\mleft[ {#1} \mright]}%
   \newcommand{\Ex}[1]{\mathop{\mathbf{E}}\pbrcx{#1}}%
   \newcommand{\Prob}[1]{\mathop{\mathbf{Pr}}\pbrcx{#1}}

   \newcommand{\cardin}[1]{\left\lvert {#1} \right\rvert}

   \newcommand{\CH}{{\mathcal{CH}}}

\newcommand{\dist}[1]{\left\| {#1}  \right\|}

\newcommand{\remlab}[1]{\label{rem:#1}}
\newcommand{\remref}[1]{Remark~\ref{rem:#1}}
\newcommand{\thmlab}[1]{{\label{theo:#1}}}
\newcommand{\thmref}[1]{Theorem~\ref{theo:#1}}
\newcommand{\lemlab}[1]{\label{lemma:#1}}
\newcommand{\lemref}[1]{Lemma~\ref{lemma:#1}}
\newcommand{\figlab}[1]{\label{fig:#1}}
\newcommand{\figref}[1]{Figure~\ref{fig:#1}}

\newcommand{\seclab}[1]{\label{sec:#1}}

\newcommand{\etal}{\textit{et~al.}\xspace}
\newcommand{\pth}[1]{\mleft(#1\mright)}%
\newcommand{\eps}{\varepsilon}%

\newcommand{\distNS}[1]{\left\| {#1} \right\|}

\newcommand{\FlatOpt}{\EuScript{F}_{\mathrm{opt}}}
\newcommand{\FlatA}{\EuScript{F}}
\newcommand{\FlatB}{\EuScript{G}}
\newcommand{\FlatC}{\EuScript{H}}
\newcommand{\FlatD}{\EuScript{I}}

\newcommand{\Ell}{\EuScript{E}}

\newcommand{\Span}{\mathrm{span}}

\newcommand{\RSample}{\EuScript{R}}
\providecommand{\Badoiu}{B\u{a}doiu\xspace}

\newcommand{\Price}[2]{\mu_{#2}\pth{#1}}
\newcommand{\PrcOpt}[2]{\mu_{\mathrm{opt}}\pth{#1, #2}}

\newcommand{\constA}{\mathsf{c}}
\newcommand{\atgen}{\symbol{'100}}

\newcommand{\SarielThanks}[1]{\thanks{Department of Computer
      Science; 
      University of Illinois; 
      201 N. Goodwin Avenue;
      Urbana, IL, 61801, USA;
      {\tt sariel\atgen{}illinois.edu}; {\tt
         \url{http://sarielhp.org/}.} #1}}

\newcommand{\DistP}[2]{\mathbf{d}_{#2}\pth{#1}}
\newcommand{\DistPSQ}[2]{\mathbf{d}_{#2}^2\pth{#1}}

\newcommand{\ApproxFlat}{\textsf{ApproxFlat}}
\newcommand{\procSVD}{\texttt{SVD}\xspace}
\renewcommand{\th}{th\xspace}

\newcommand{\NSize}{\mathsf{N}{}}

\newcommand{\MatA}{\EuScript{M}}
\newcommand{\MatB}{\EuScript{B}}
\newcommand{\MatC}{\EuScript{C}}

\newcommand{\Fnorm}[1]{\left\| {#1} \right\|_F}

\newcommand{\ConstSL}{96}
\newcommand{\ConstSAlg}{500}

\newcommand{\ceil}[1]{\left\lceil {#1} \right\rceil}

\newcommand{\Term}[1]{\textsf{#1}}
\newcommand{\SVD}{\Term{SVD}\xspace}%

\theoremstyle{remark}%
\theoremheaderfont{\sf} \theorembodyfont{\upshape}%

   \newtheorem*{remark:unnumbered}[theorem]{Remark}%
   \newtheorem{remark}[theorem]{Remark}%

   \theoremheaderfont{\em}%
   \theorembodyfont{\upshape}%
   \theoremstyle{nonumberplain}
   \theoremseparator{}
   \theoremsymbol{\myqedsymbol}
   \newtheorem{proof}{Proof:}

\begin{document}

\title{Low Rank Matrix Approximation in Linear Time}

\author{Sariel Har-Peled\SarielThanks{Work on this paper
      was partially supported by a NSF CAREER award
      CCR-0132901.}}

\date{January 24, 2006}

\maketitle

\begin{abstract}
    Given a matrix $\MatA$ with $n$ rows and $d$ columns, and fixed
    $k$ and $\eps$, we present an algorithm that in linear time (i.e.,
    $O(\NSize )$) computes a $k$-rank matrix $\MatB$ with
    approximation error $\Fnorm{\MatA - \MatB}^2 \leq (1+\eps)
    \PrcOpt{\MatA}{k}$, where $\NSize = n d$ is the input size, and
    $\PrcOpt{\MatA}{k}$ is the minimum error of a $k$-rank
    approximation to $\MatA$.

    This algorithm succeeds with constant probability, and to our
    knowledge it is the first linear-time algorithm to achieve
    multiplicative approximation.
\end{abstract}


\section{Introduction}

In this paper, we study the problem of computing a low rank
approximation to a given matrix $\MatA$ with $n$ rows and $d$ columns.
A $k$-rank approximation matrix is a matrix of rank $k$ (i.e., the
space spanned by the rows of the matrix is of dimension $k$). The
standard measure of the quality of approximation of a matrix $\MatA$
by a matrix $\MatB$ is the squared Frobenius norm $\Fnorm{\MatA -
   \MatB}^2$, which is the sum of the squared entries of the matrix
$\MatA - \MatB$.

The optimal $k$-rank approximation can be computed by \SVD in time $O(
\min( n d^2, d n^2 ) )$. However, if $d$ and $n$ are large
(conceptually, consider the case $d=n$), then this is unacceptably
slow. Since one has to read the input at least once, a running time of
$\Omega(nd)$ is required for any matrix approximation algorithm. Since
$\NSize = n d$ is the size of the input, we will refer to running time
of $O(\NSize)$ as being linear.

\paragraph{Previous work.}
Frieze \etal{} \cite{fkv-fmcaf-04} showed how to sample the rows of
the matrix, so that the resulting sample span a matrix which is a good
low rank approximation. The error of the resulting approximation has
an additive term that depends on the norm of the input matrix.
Achlioptas and McSherry \cite{am-fclra-01} came up with an alternative
sampling scheme that yields a somewhat similar result.
 
Note, that a multiplicative approximation error which is proportional
to the error of the optimal approximation is more desirable as it is
potentially significantly smaller than the additive error. However,
computing quickly a small multiplicative low-rank approximation
remained elusive.  Recently, Deshpande \etal{} \cite{drvw-mapcv-06},
building on the work of Frieze \etal{} \cite{fkv-fmcaf-04}, made a
step towards this goal. They showed a multipass algorithm with an
additive error that decreases exponentially with the number of passes.
They also introduced the intriguing concept of volume sampling.

\Badoiu{} and Indyk \cite{bi-faahf-06} had recently presented a linear
time algorithm for this problem (i.e., $O( \NSize + \pth{d + \log
   n}^{O(1)} )$), for the special case $k=1$, where $\NSize = n d$.
They also mention that, for a fixed $k$, a running time of $O( \NSize
k \log(n/\eps) \log(1/\eps) )$ is doable.  In fact, using the Lanczos
method or Power method yields a $(1+\eps)$-approximation in time
roughly $O(\NSize (k/\eps)^{O(1)} \log n )$.  For $k=1$, this was
proved in the work of Kuczynski and Wozniakowski \cite{kw-elepl-92},
and the proof holds in fact for any fixed $k$.  We also point out
another algorithm, with similar performance, using known results, see
\remref{slow:but:easy}. Thus, the challenge in getting a linear
running time is getting rid of the $O(\log n)$ factor in the running
time.

In practice, the Power method requires $O(\NSize k I)$ time, where $I$
is the number of iterations performed, which depends on the
distribution of the eigenvalues, and how close they are to each other.
In fact, the $\NSize$ term can be replaced by the number of non-zero
entries in the matrix. In particular, if the matrix is sparse (and $I$
is sufficiently small) the running time to compute a low-rank
approximation to the matrix is sublinear.

One can interpret each row of the matrix $\MatA$ as a point in
$\Re^d$. This results in a point-set $P$ of $n$ points in $\Re^d$. The
problem now is to compute a $k$-flat (i.e., a $k$-dimensional linear
subspace) such that the sum of the squared distances of the points of
$P$ to the $k$-flat is minimized. This problem is known as the
$L_2$-fitting problem. The related $L_1$-fitting problem was recently
studied by Clarkson \cite{c-ssalo-05} and the $L_\infty$-fitting
problem was recently studied by Har-Peled and Varadarajan
\cite{hv-hdsfl-04} (see references therein for the history of these
problems).

\paragraph{Our Results.}
The extensive work on this problem mentioned above still leaves open
the nagging question of whether one can get a multiplicative low-rank
approximation to a matrix in linear time (which is optimal in this
case).  In this paper, we answer this question in the positive and
show a linear-time algorithm with small multiplicative error, using a
more geometric interpretation of this problem. In particular, we
present an algorithm that with constant probability computes a matrix
$\MatB$ such that $\Fnorm{\MatA- \MatB}^2 \leq (1+\eps)
\PrcOpt{\MatA}{k}$, where $\PrcOpt{P}{k}$ is the minimum error of a
$k$-rank approximation to $\MatA$, and $\eps > 0$ is a prespecified
error parameter. For input of size $\NSize = n d$ the running time of
the new algorithm is (roughly) $O(\NSize k^2 \log k)$; see
\thmref{main} for details.

Our algorithm relies on the observation that a small random sample
spans a flat which is good for most of the points. This somewhat
intuitive claim is proved by using $\eps$-nets and VC-dimension
arguments performed (conceptually) on the optimal (low-dimensional)
$k$-flat (see \lemref{r:sample} below). Next, we can filter the
points, and discover the outliers for this random flat (i.e., those
are the points which are ``far'' from the random flat).  Since the
number of outliers is relatively small, we can approximate them
quickly using recursion.  Merging the random flat together with the
flat returned from the recursive call results in a flat that
approximates well all the points, but its rank is too high. This can
be resolved by extracting the optimal solution on the merged flat, and
improving it into a good approximation using the techniques of Frieze
\etal{} \cite{fkv-fmcaf-04} and Deshpande \etal{}
\cite{drvw-mapcv-06}.

\section{Preliminaries}

\subsection{Notations}

A $k$-\emphi{flat} $\FlatA$ is a linear subspace of $\Re^d$ of
dimension $k$ (in particular, the origin is included in $\FlatA$).

For a point $q$ and a flat $\FlatA$, let $\DistP{q}{\FlatA}$ denote
the \emphi{distance} of $q$ to $\FlatA$. Formally, $\DistP{q}{\FlatA} =
\min_{x \in \FlatA} \dist{ q -x}$. We will denote the squared distance
by $\DistPSQ{q}{\FlatA} = \pth{\DistP{q}{\FlatA}}^2$.
    
The \emphi{price} of fitting a point set $P$ of $n$ points in $\Re^d$
to a flat $\FlatA$ is 
\begin{align*}
    \Price{P}{\FlatA} = \sum_{p \in P} \pth{\DistP{p}{\FlatA}}^2.
\end{align*}
The $k$-flat realizing $\displaystyle \min_{\dim(\FlatA) = k}
\Price{P}{\FlatA}$ is denoted by $\FlatOpt$, and its price is
$\PrcOpt{P}{k}$.

The \emphi{span} of a set $S \subseteq \Re^d$, denoted by $\Span(S)$,
is the smallest linear subspace that contains $S$.

The input is a set $P$ of $n$ points in $\Re^d$ and its size is
$\NSize = n  d$.


\subsection{Computing a $(s,\FlatA)$-sample.}

Let $\FlatA$ be a flat.  For a parameter $s$, a \emphi{$(s,
   \FlatA)$-sample} from $P$ is generated by picking point $p\in P$ to
the random sample with probability
\begin{align*}
\geq \constA \frac{\pth{\DistPSQ{p}{\FlatA}} }
{\sum_{q \in P} \pth{\DistPSQ{q}{\FlatA}} },
\end{align*}
where $\constA$ is an appropriate constant, and the size of the sample
is at least $s$.

It is not immediately clear how to compute a $(s,\FlatA)$-sample
efficiently since we have to find the right value of $\constA$ such
that the sample generated is of the required size (and not much bigger
than this size). However, as described by Frieze \etal{}
\cite{fkv-fmcaf-04}, this can be achieved by bucketing the basic
probabilities (i.e., ${{\DistPSQ{p}{\FlatA}} } /\sum_{q \in P}
{\DistPSQ{q}{\FlatA}}$ for a point $p \in P$), such that each bucket
corresponds to points with probabilities in the range
$[2^{-i},2^{-i+1}]$. It is now easy to verify that given a target size
$s$, one can compute the corresponding $\constA$ in linear time, such
that the expected size of the generated sample is, say, $4s$.  We also
note, that by the Chernoff inequality, with probability $\geq 1/2$,
the sample is of size in the range $[s,8s]$.

\begin{lemma}
    \lemlab{weighted:sample}%
    Given a flat $\FlatA$ and a parameter $s$, one can compute a
    $(s,\FlatA)$-sample from a set $P$ of $n$ points in $\Re^d$ in $O(
    \NSize \dim(\FlatA) )$ time. The sample is of size in the range
    $[s,8s]$ with probability $\geq 1/2$.

    Furthermore, if we need to generate $t$ independent
    $(s,\FlatA)$-samples, all of size in the range $[s,8s]$, then this
    can be done in $O( \NSize \dim(\FlatA) + n \cdot t )$ time, and
    this bound on the running time holds with probability $\geq 1-
    \exp(-t)$.
\end{lemma}

\begin{proof}
    The first claim follows from the above description. As for the
    second claim, generate $16 t$ samples independently. Note, that we
    need to compute the distance of the points of $P$ to $\FlatA$ only
    once before we compute these $16t$ samples. By the Chernoff
    inequality, the probability that less than $t$ samples, out of the
    $16 t$ samples generated, would be in the valid size range is
    smaller than $\exp(-t)$.
\end{proof}

\subsection{Extracting the best $k$-flat lying inside %
   another flat}

\begin{lemma}
    \lemlab{flat:extract}%
    Let $\FlatA$ be a $u$-flat in $\Re^d$, and let $P$ be a set of $n$
    points in $\Re^d$. Then, one can extract the $k$-flat
    $\FlatB\subseteq \FlatA$ that minimizes $\Price{P}{\FlatB}$ in $O(
    \NSize u + n u^2)$ time, where $k \leq u$.
\end{lemma}
\begin{proof}
    We project the points of $P$ onto $\FlatA$. This takes $O( \NSize
    u)$ time, and let $P'$ denote the projected set. Next, we compute,
    using \SVD, the best approximation to $P'$ by a $k$-flat $\FlatB$
    inside $\FlatA$. Since $\FlatA$ is $u$ dimensional, we can treat
    the point of $P'$ as being $u$ dimensional and this can be done in
    $O(nu^2)$ time overall.

    For a point $p \in P$, let $p'$ be its projection onto $\FlatA$
    and let $p''$ be the projection of $p$ onto $\FlatB$. By
    elementary argumentation, we know that the projection of $p'$ onto
    $\FlatB$ is $p''$. As such, $\dist{p - p''}^2 = \dist{p - p'}^2 +
    \distNS{p' - p''}^2$ implying that $\Price{P}{\FlatB} =
    \Price{P}{\FlatA} + \Price{P'}{\FlatB}$. Since we computed the
    flat $\FlatB$ that minimizes $\Price{P'}{\FlatB}$, it follows that
    we computed the $k$-flat $\FlatB$ lying inside $\FlatA$ that
    minimizes the approximation price $\Price{P}{\FlatB}$.
\end{proof}

\subsection{Fast additive approximation}
\seclab{fast:additive}

We need the following result (the original result of
\cite{drvw-mapcv-06} is in fact slightly stronger).
\begin{theorem}
    \thmlab{sample:a}%
    Let $P$ be a set of $n$ points in $\Re^d$,  $\FlatA$ be a
    $u$-flat, $\RSample$ be a $(s,\FlatA)$-sample of size $r$, and
    let $\FlatB = \Span(\FlatA \cup \RSample)$. Let $\FlatC$ be the
    $k$-flat lying inside $\FlatB$ that minimizes $\Price{P}{\FlatC}$.
    Then, we have that
    \begin{align*}
    \Ex{ \Price{P} {\FlatC} } \leq \PrcOpt{P}{k} + \frac{k}{r}
    \Price{P}{\FlatA}.
    \end{align*}
    Furthermore, $\FlatC$ can be computed in $O( \NSize (r + u) )$
    time. 
\end{theorem}

\begin{proof}
    The correctness follows immediately from the work of Deshpande
    \etal{} \cite{drvw-mapcv-06}.

    As for the running time, we compute the distance of the points of
    $P$ to $\FlatA$. This takes $O( \NSize u )$ time.  Next, we
    compute $\FlatB = \Span(\FlatA \cup \RSample)$.  Finally, we
    compute the best $k$-flat $\FlatC \subseteq \FlatB$ that
    approximates $P$ using \lemref{flat:extract}, which takes
    $O(\NSize (r+u) + n(r+u)^2) = O(\NSize (r+u))$ since $r+u \leq
    2d$.
\end{proof}

We need a probability guaranteed version of \thmref{sample:a}, as
follows.
\begin{lemma}
    \lemlab{guaranteed}%
    Let $P$ be a set of $n$ points in $\Re^d$, $\FlatA$ be a
    $O(k)$-flat, and parameters $\delta >0$ and $\eps > 0$. Then, one
    can compute a $k$-flat $\FlatC$, such that, with probability at
    least $\geq 1-\delta$, we have $\Price{P} {\FlatC} \leq \beta$,
    where
    \begin{align*}    
    \beta = (1+\eps/2)\PrcOpt{P}{k} + (\eps/2) 
    \Price{P}{\FlatA}.
    \end{align*}
    The running time of the algorithm is $O( \NSize (k/\eps) \log
    (1/\delta))$. The algorithm succeeds with probability $\geq
    1-\delta$.
\end{lemma}

\begin{proof}
    Set $s = 4k/\eps$. Generate $t=O( (1/\eps) \ln (1/\delta))$
    independent $(s,\FlatA)$-samples, of size in the range $[s,8s]$,
    using the algorithm of \lemref{weighted:sample}.  Next, for each
    random sample $\RSample_i$, we apply the algorithm of
    \thmref{sample:a} to it, generating a $k$-flat $\FlatC_i$, for
    $i=1,\ldots, t$.  By the Markov inequality, we have
    \begin{align*}
    \Prob{ \Bigl. \Price{P} {\FlatC_i} \geq (1+\eps/4) \Ex{ \Price{P}
          {\FlatC_i} } } \leq \frac{1}{1+\eps/4}.
    \end{align*}
    As such, with probability $\geq 1 - 1/(1+\eps/4) \geq 1 -
    (1-\eps/8) = \eps/8$ we have
    \begin{eqnarray*}
        \alpha_i = \Price{P} {\FlatC_i} &\leq& 
        \pth{1 + \frac{\eps}{4}} \pth{
           \PrcOpt{P}{k} + \frac{k}{s}
           \Price{P}{\FlatA}} \\
        &= &
        \pth{1 + \frac{\eps}{4}}\PrcOpt{P}{k} +
        \pth{1 + \frac{\eps}{4}} \cdot
        \frac{\eps}{4}
        \Price{P}{\FlatA} \\
        &\leq& \beta.
    \end{eqnarray*}
    Clearly, the best $k$-flat computed, which realizes the price
    $\displaystyle \min_{i=1}^t \alpha_i$, has a price which is at
    most $\beta$, and this holds with probability $\geq 1-
    (1-\eps/8)^t \geq 1- \delta/2$.  Note, that the time to generate
    the samples is $O( \NSize \dim(\FlatA) + n \cdot t)$, and this
    bounds holds with probability $\geq 1-\delta/4$. As such, the
    overall running time of the algorithm is $O( \NSize \dim(\FlatA) +
    n \cdot t + \NSize k t ) = O( \NSize (k/\eps) \ln (1/\delta))$.
\end{proof}

\begin{lemma}
    \lemlab{easy}%
    Let $\FlatA$ be a given $O(k)$-flat, $P$ a set of points $n$ in
    $\Re^d$, and $c$ a given parameter. Assume that $\Price{P}{\FlatA}
    \leq c \cdot \PrcOpt{P}{k}$. Then, one can compute a $k$-flat
    $\FlatB$ such that
    \begin{align*}
    \Price{P}{\FlatB} \leq (1+\eps)\PrcOpt{P}{k}.
    \end{align*}
    The running time of the algorithm is $O( \NSize \cdot (k/\eps)
    \log (c/\delta))$ and it succeeds with probability $\geq
    1-\delta$.
\end{lemma}

\begin{proof}
    Let $t= 16 \ceil{ \lg (c/\delta) }$, and set $\FlatA_0 = \FlatA$.
    In the $i$\th step of the algorithm, we generate a
    $(s,\FlatA_{i-1})$-sample $\RSample_i$ with $s=20k$.  If the
    generated sample is not of size in the range $[s,8s]$, we
    resample.  Next, we apply \thmref{sample:a} to $\FlatA_{i-1}$ and
    $\RSample_i$.  If the new generated flat $\FlatB_i$ is more
    expensive than the $\FlatA_{i-1}$ then we set $\FlatA_{i} =
    \FlatA_{i-1}$.  Otherwise, we set $\FlatA_{i} = \FlatB_i$. We
    perform this improvement loop $t$ times.

    Note, that if $\Price{\FlatA_{i-1}}{P} \geq 8 \PrcOpt{P}{k}$ then
    \begin{align*}
    \Ex{ \Price{P} {\FlatB_{i}} } \leq \PrcOpt{P}{k} + \frac{k}{20k}
    \Price{P}{\FlatA_{i-1}} \leq \pth{ \frac{1}{8} +
       \frac{1}{20}}\Price{P}{\FlatA_{i-1}} \leq
    \frac{1}{4}\Price{P}{\FlatA_{i-1}},
    \end{align*}
    by \thmref{sample:a}.  By Markov's inequality, it follows that
    $\Prob{ \Price{P} {\FlatB_{i}} \leq
       \frac{1}{2}\Price{P}{\FlatA_{i-1}}} \geq 1/2$. Namely, with
    probability half, we shrink the price of the current $k$-flat by a
    factor of $2$ (in the worst case the price remains the same).
    Since each iteration succeeds with probability half, and we
    perform $t$ iterations, it follows that with probability larger
    than $1-\delta/10$ (by the Chernoff inequality), the price of the
    last flat $\FlatA_t$ computed is $\leq 8 \PrcOpt{P}{k}$.
    Similarly, we know by the Chernoff inequality that the algorithm
    performed at most $4t$ sampling rounds till it got the required
    $t$ good samples, and this holds with probability $\geq
    1-\delta/10$.

    Thus, computing $\FlatA_t$ takes $O( \NSize{} k t ) = O( \NSize \,
    k \log (c/\delta))$ time.  Finally, we apply \lemref{guaranteed}
    to $\FlatA_t$ with $\eps/10$ and success probability $\geq 1-
    \delta/10$. The resulting $k$-flat has price
    \begin{eqnarray*}
        \Price{P}{\FlatC} &\leq& (1+\eps/20)\PrcOpt{P}{k} + (\eps/20)
        \Price{P}{\FlatA_t} \leq  (1+\eps/20 + 8(\eps/20))\PrcOpt{P}{k}\\
        &\leq& ( 1+\eps) \PrcOpt{P}{k},
    \end{eqnarray*}
    as required. Note, that overall this algorithm succeeds with
    probability $\geq 1-\delta$.
\end{proof} 

\begin{remark}
    \remlab{slow:but:easy}%
    Har-Peled and Varadarajan \cite{hv-hdsfl-04} presented an
    algorithm that computes, in $O(\NSize k)$ time, a $k$-dimensional
    affine subspace (i.e., this subspace does not necessarily contains
    the origin) which approximates, up to a factor of $k!$, the best
    such subspace that minimizes the maximum distance of a point of
    $P$ to such a subspace. Namely, this is an approximation algorithm
    for the $L_\infty$-fitting problem, and let $\FlatA$ denote the
    computed $(k+1)$-flat (i.e., we turn the $k$-dimensional affine
    subspace into a $(k+1)$-flat by adding the origin to it).

    It is easy to verify that $\FlatA$ provides a $O\pth{ (n k!)^2 }$
    approximation to the optimal $k$-flat realizing $\PrcOpt{P}{k}$
    (i.e., the squared $L_2$-fitting problem).  To see that, we remind
    the reader that the $L_\infty$ and $L_2$ metrics are the same up
    to a factor which is bounded by the dimension (which is $n$ in
    this case).

    We can now use \lemref{easy} starting with $\FlatA$, with $c=
    O\pth{ (n k!)^2 }$. We get the required $(1+\eps)$-approximation,
    and the running time is $O( \NSize (k/\eps) (k \log k + \log
    (n/\delta)))$.  The algorithm succeeds with probability $\geq
    1-\delta$.
\end{remark}

\subsection{Additional tools}

\begin{lemma}
    \lemlab{distance:2}%
    Let $\DistPSQ{x}{\FlatA} =\min_{x \in \FlatA} \dist{ q -x}^2$ be
    the squared distance function of a point $x$ to a given $u$-flat
    $\FlatA$. Let $x, y \in \Re^d$ be any two points. We have:
    \begin{itemize}
        \item[(i)] For any $0 \leq \beta \leq 1$, $\DistPSQ{\beta x +
           (1 - \beta) y}{\FlatA} \leq \beta \DistPSQ{x}{\FlatA} + (1
        - \beta) \DistPSQ{y}{\FlatA}$. That is,
        $\DistPSQ{\cdot}{\FlatA}$ is convex.
        
        \item[(ii)] Let $\ell$ be the line through $x$ and $y$, and
        let $z$ be a point on $\ell$ such that $\distNS{x - z} \leq k
        \cdot \distNS{x - y}$. Then,
        \begin{align*}
        \DistPSQ{z}{\FlatA} \leq 9 k^2 \cdot \max \pth{
           \DistPSQ{x}{\FlatA}, \DistPSQ{y}{\FlatA}}.
        \end{align*}
    \end{itemize}
\end{lemma}

\begin{proof}
    The claims of this lemma follow by easy elementary arguments and
    the proof is included only for the sake of completeness.

    (i) Rotate and translate space, such that $\FlatA$ becomes the
    $u$-flat spanned by the first $u$ unit vectors of the standard
    orthonormal basis of $\Re^d$, denoted by $\FlatA'$. Let $x'$ and
    $y'$ be the images of $x$ and $y$, respectively. Consider the line
    $\ell(t) = t x' + (1-t) y'$. We have that $g(t) =
    \DistPSQ{\ell(t)}{\FlatA'} = \sum_{i=t+1}^d (t x_i' + (1-t)
    y_i')^2$, which is a convex function since it is a sum of convex
    functions, where $x' = (x_1',\ldots, x_d')$ and $y' =
    (y_1',\ldots, y_d')$. As such, $g(0) = \DistPSQ{y}{\FlatA}$ and
    $g(1) = \DistPSQ{x}{\FlatA}$, and $g(\beta) = \DistPSQ{\beta x +
       (1 - \beta) y}{\FlatA}$. By the convexity of $g(\cdot)$, we
    have $\DistPSQ{\beta x + (1 - \beta) y}{\FlatA} = g(\beta) \leq
    \beta g(1) + (1-\beta) g(0) = \beta \DistPSQ{x}{\FlatA} + (1 -
    \beta)\DistPSQ{y} {\FlatA}$, as claimed.

    \parpic[r]{ \begin{minipage}{5.2cm} \includegraphics{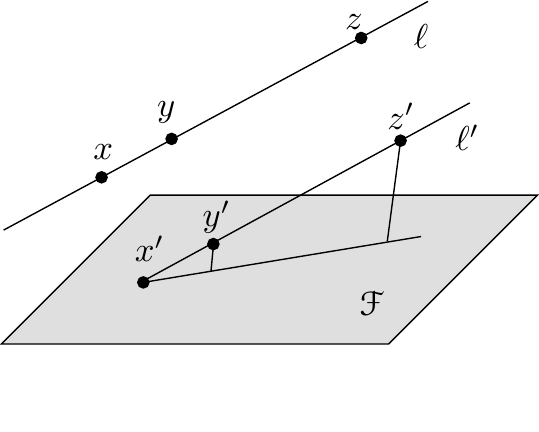}
           \vspace*{-4cm}
       \end{minipage}
    }
    (ii)  Let ${x}'$ be the projection of $x$ to $\FlatA$, and
    let $\ell'$ be the parallel line to $\ell$ passing through
    ${x}'$. Let $y' \in \ell'$ and $z'\in \ell'$ be the
    translation of $y$ and $z$ onto $\ell'$ by the vector $x'-x$. By
    similarity of triangles, we have
    \begin{align*}
    \DistP{z'}{\FlatA} = \frac{\dist{x' - z'}}{\dist{x'-y'}}
    \DistP{y'}{\FlatA} \leq k\cdot \DistP{y'}{\FlatA}.
    \end{align*}

    Thus,
    \begin{align*}
        \displaystyle \DistPSQ{z}{\FlatA}%
        &\leq%
        \pth{\dist{x' - x} + \DistP{z'}{\FlatA}}^2 \leq%
        \pth{\DistP{x}{\FlatA} + k \cdot \DistP{y'}{\FlatA}}^2%
        \\&%
        \leq%
        \pth{(k+1)\DistP{x}{\FlatA} + k \cdot \DistP{y}{\FlatA}}^2%
        \leq%
        9 k^2 \cdot \max \pth{ \DistPSQ{x}{\FlatA},
           \DistPSQ{y}{\FlatA}},
    \end{align*}
    since $\DistP{y'}{\FlatA} \leq \dist{ x- x'} + \DistP{y}{\FlatA} =
    \DistP{x}{\FlatA} + \DistP{y}{\FlatA}$.
\end{proof}

        


\noindent%
\begin{minipage}{0.65\linewidth}%
    \begin{lemma}
    \lemlab{extract:k:flat}%
    Let $P$ be a set of $n$ points in $\Re^d$, and let $\FlatA$ be a
    flat in $\Re^d$, such that $\Price{P}{\FlatA} \leq c \cdot
    \PrcOpt{P}{k}$, where $c\geq 1$ is a constant. Then, there exists
    a $k$-flat $\FlatB$ that lies inside $\FlatA$ such that
    $\Price{P}{\FlatB} \leq 5 c \cdot \PrcOpt{P}{k}$.
\end{lemma}
\end{minipage}%
\hfill
\begin{minipage}{0.3\linewidth}%
    \includegraphics[width=1.01\linewidth]{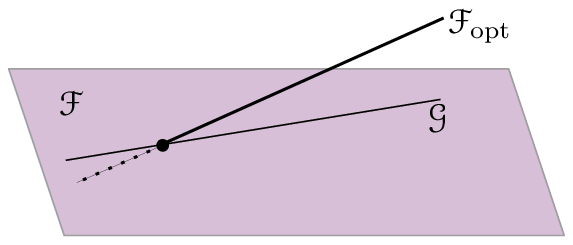}%
    \vspace*{-0.7cm}%
\end{minipage}%

\begin{proof}
    Let $\FlatB$ be the projection of $\FlatOpt$ onto $\FlatA$, where
    $\FlatOpt$ denotes the optimal $k$-flat that approximates $P$.
    Let $P'$ be the projection of $P$ onto $\FlatA$. We have that
    $\Price{P'}{\FlatB} \leq \Price{P'}{\FlatOpt}$. On the other hand,
    \begin{eqnarray*}
        \Price{P}{\FlatB} &=& \sum_{p \in P} \pth{ \pth{\DistP{p}{\FlatA}}^2
           + \pth{\DistP{p'}{\FlatB}}^2 } = \Price{P}{\FlatA} +
        \Price{P'}{\FlatB}
        \leq \Price{P}{\FlatA} + \Price{P'}{\FlatOpt},
    \end{eqnarray*}
    where $p'$ denotes the projection of a point $p \in P$ onto
    $\FlatA$. Furthermore,
    \begin{eqnarray*}
        \Price{P'}{\FlatOpt} &=& \sum_{p \in P} 
        \pth{\DistP{p'}{\FlatOpt} }^2 
        \leq \sum_{p \in P} 
        \pth{\dist{ p - p'} + \DistP{p}{\FlatOpt} }^2
        \leq 2 \sum_{p \in P'}\pth{ \pth{ \DistP{p}{\FlatA}}^2 + 
           \pth{ \DistP{p}{\FlatOpt} }^2}     \\
        &\leq & 2 \pth{\Price{P}{\FlatA} + \Price{P}{\FlatOpt}} 
        \leq (2c+2) \Price{P}{\FlatOpt}.
    \end{eqnarray*}
    We conclude that $\Price{P}{\FlatB} \leq \Price{P}{\FlatA} +
    \Price{P'}{\FlatOpt} \leq (c + 2c+2) \Price{P}{\FlatOpt} \leq 5c
    \cdot \PrcOpt{P}{k}$.
\end{proof}

\section{A sampling lemma}

The following lemma, testifies that a small random sample induces a
low-rank flat that is a good approximation for most points. The random
sample is picked by (uniformly) choosing a subset of the appropriate
size from the input set, among all subsets of this size.
\begin{lemma}
    \lemlab{r:sample}%
    Let $P$ be a set of $n$ points in $\Re^d$, $\delta >0$ a
    parameter, and let $\RSample$ be a random sample of $P$ of size
    $O\pth{ k^2 \log (k/\delta) }$. Consider the flat $\FlatA =
    \Span(\RSample)$, and let $Q$ be the set of the $(3/4)n$ points of
    $P$ which are closest to $\FlatA$.  Then $\Price{Q}{\FlatA} \leq
    \ConstSL k^2 \cdot \PrcOpt{P}{k}$ with probability $\geq
    1-\delta$.
\end{lemma}

\begin{proof}
    Let $r = \PrcOpt{P}{k}/n$ be the average contribution of a point
    to the sum $\Price{P}{\FlatOpt}= \sum_{p\in P}
    \DistPSQ{p}{\FlatOpt}$, where $\FlatOpt$ is the optimal $k$-flat
    approximating $P$.  Let $U$ be the points of $P$ that contribute
    at most $8r$ to $\Price{P}{\FlatOpt}$; namely, those are the
    points in distance at most $\sqrt{8r}$ from $\FlatOpt$.  By
    Markov's inequality, we have $\cardin{U} \geq (7/8)n$.  Next,
    consider the random sample when restricted to the points of $U$;
    namely, $R = \RSample \cap U$.  By the Chernoff inequality, with
    probability at least $1- \delta/4$, we have that $\cardin{R} =
    \Omega(k^2 \log (k/\delta))$.
    
    Next, consider the projections of $R$ and $U$ onto $\FlatOpt$,
    denoted by $R'$ and $U'$, respectively.  Let $\Ell'$ be the
    largest volume ellipsoid, contained in $\FlatOpt$ that is enclosed
    inside $\CH(R')$, where $\CH(R')$ denotes the convex-hull of $R'$.
    Let $\Ell$ be the expansion of $\Ell'$ by a factor of $k$. By
    John's theorem \cite{m-ldg-02}, we have $R' \subseteq \CH(R')
    \subseteq \Ell$.

    Since $R'$ is a large enough random sample of $U'$, we know that
    $R'$ is a $(1/24)$-net for ellipsoids for the set $U'$. This
    follows as ellipsoids at Euclidean space $\Re^k$ have VC-dimension
    $O(k^2)$, as can be easily verified. (One way to
    see this is by lifting the points into $O(k^2)$ dimensions, where
    an ellipsoid in the original space is mapped into a half space.)

    \piccaption{The solid discs represent point in $R'$, the
       circle represents the points of $U'$.}
    \parpic[r]{
       \begin{minipage}{7.2cm}
           \includegraphics[width=0.99\linewidth]{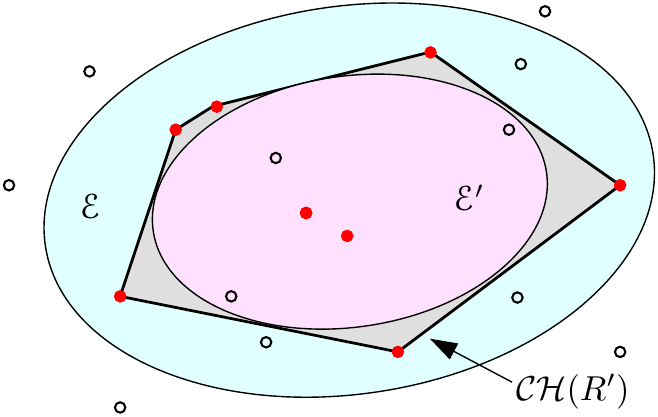}
       \end{minipage}
    }%
    As such, we can use the $\eps$-net theorem of Haussler and Welzl
    \cite{hw-ensrq-87}, which implies that $\cardin{ U' \cap \Ell}
    \geq (23/24)\cardin{U'}$, with probability at least $1-\delta/4$.
    This holds since no points of $R'$ falls outside $\Ell$, and as
    such only $(1/24)\cardin{U'}$ points of $U'$ might fall outside
    $\Ell$ by the $\eps$-net theorem. (Strictly speaking, the
    VC-dimension argument is applied here to ranges that are
    complements of ellipsoids. Since the VC dimension is persevered
    under complement the claim still holds.)

    For $p' \in \FlatOpt$, let $f(p') = \DistPSQ{p'}{\FlatA}$ denote
    the squared distance of $p'$ to $\FlatA$.  If $p' \in R'$ then the
    corresponding original point $p \in R \subseteq U$ and as such $p
    \in \FlatA = \Span(R)$. Implying that $f(p') =
    \DistPSQ{p'}{\FlatA} \leq \dist{p - p'}^2 = \DistPSQ{p}{\FlatOpt}
    \leq 8r$ since $p \in U$.  By \lemref{distance:2} (i), the
    function $f(\cdot)$ is convex and, for any $x \in \Ell'$, we have
    \begin{align*}
    f(x) \leq \max_{y \in \Ell'} f(y) \leq \max_{z \in \CH(R')} f(z)
    \leq \max_{w \in R'} f(w) \leq 8r,
    \end{align*}
    since $\Ell' \subseteq \CH(R')$.

    \parpic[l]{
       \begin{minipage}{4.2cm}
           \includegraphics[width=0.99\linewidth]{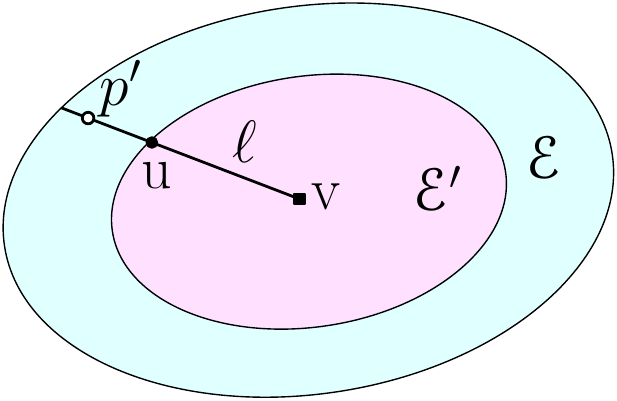}
       \end{minipage}
    } 

    Let $v$ denote the center of $\Ell'$, and consider any point $p'
    \in \Ell \subseteq \FlatOpt$. Consider the line $\ell$ connecting
    $v$ to $p'$, and let $u$ be one of the endpoints of $\ell \cap
    \Ell'$.  Observe that $\distNS{\ell \cap \Ell} \leq k \cdot \distNS{\ell
       \cap \Ell'}$. This implies by \lemref{distance:2} (ii) that
    \begin{eqnarray*}
        \forall p' \in \Ell \;\; \;\; 
        f(p') = \DistPSQ{p'}{\FlatA} \leq 9 k^2  \max(\DistP{u}{\FlatA},
        \DistP{v}{\FlatA}) \leq 72 k^2 r.
    \end{eqnarray*}
    In particular, let $Y$ be the set of points of $U$ such that their
    projection onto $\FlatOpt$ lies inside $\Ell$. For $p \in Y$, we
    have that
    \begin{align*}
    \DistPSQ{p}{\FlatA} \leq \pth{ \distNS{p - p'} +
       \DistP{p'}{\FlatA}}^2 = \pth{ \sqrt{8r} + \sqrt{72k^2 r}}^2
    \leq 128k^2 r,
    \end{align*}
    where $p'$ is its the projection of $p$ onto $\FlatOpt$.
    Furthermore, $\cardin{Y} = \cardin{U' \cap \Ell} \geq (23/24)
    \cardin{U'} \geq (23/24)(7/8) 8 \geq (3/4) n$.

    Specifically, let $Q$ be the set of $(3/4)n$ points of $P$
    closest to $\FlatA$. Since $\cardin{Q} \leq \cardin{Y}$, we have
    that for any $q \in Q$, it holds $\DistPSQ{q}{\FlatA} \leq \max_{y
       \in Y} \DistPSQ{y}{\FlatA} \leq 128 k^2 r$ and as such
    $\Price{Q}{\FlatA} \leq (3/4)n \cdot 128 k^2 r = \ConstSL k^2
    \cdot \PrcOpt{P}{k}$.
\end{proof}


\section{The algorithm}

\begin{figure}
    \centerline{\fbox{
           \begin{minipage}{0.9\linewidth}
               \begin{tabbing}
                   \ \ \ \= \ \ \ \= \ \ \ \=  \ \ \ \= \kill
                   \ApproxFlat{}($P, k, \delta, \eps$ ) \+ \+ \\
                   $P$ - point set, \ \ \ 
                   $k$ - dimension of required flat, \\
                   $\delta$ - confidence, \ \ 
                   $\eps$ - quality of approximation. \-\-\\
                   \Procbegin{} \+ \+ \\
                   \If $\cardin{P} \leq 20k$ \Then \Return $procSVD{}(P,k)$\\
                   $\RSample \leftarrow$ random sample
                   of size $ O\pth{ k^2 \log (k/\delta) }$
                   from $P$
                   (using \lemref{r:sample})\\
                   $\FlatA = \Span(\RSample)$\\
                   $X \leftarrow$ The set of the $(3/4)n$ closest points 
                   of $P$ to $\FlatA$, \ \ \ and \ \ \ 
                   $Y \leftarrow P \setminus X$.\\
                   $\FlatB = \ApproxFlat(Y, k,
                   \delta/2, 1)$.\\
                   $\FlatC \leftarrow$ extract best $k$-flat
                   in $\Span(\FlatA \cup \FlatB)$
                   approximating $P$ 
                   (using \lemref{flat:extract})\\
                   $\FlatD \leftarrow$ Compute a $k$-flat of price
                   $\leq (1+\eps)\PrcOpt{P}{k}$ \\
                   \ \ \ \ \ \ \ \ \ \ 
                   (using \lemref{easy} with $\FlatC$ and  $c =
                   \ConstSAlg k^2/\eps$)\\
                   \Return $\FlatD$ \- \- \\
                   \End{}
               \end{tabbing}
           \end{minipage}}}
     \caption{The algorithm for approximating the best $k$. The
        procedure \procSVD$(P,k)$ extracts the best $k$-flat using
        singular value decomposition computation.  The algorithm
        returns a $k$-flat $\FlatD$ such that $\Price{P}{\FlatD} \leq
        (1+\eps)\PrcOpt{P}{k}$, and this happens with probability
        $\geq 1-\delta$.}
    \figlab{algorithm}
\end{figure}

The new algorithm \ApproxFlat{} is depicted in \figref{algorithm}.
\begin{theorem}
    \thmlab{main}%
    For a set $P$ of $n$ points in $\Re^d$, and parameters $\eps$ and
    $\delta$, the algorithm \ApproxFlat{}$(P,k,\delta,\eps)$ computes
    a $k$-flat $\FlatD$ of price $\Price{P}{\FlatD} \leq
    (1+\eps)\PrcOpt{P}{k}$. If $k^2 \log(k/\delta) \leq d$, the
    running time of this algorithm is $O \pth{\NSize k \pth{ \eps^{-1}
          + k } { \log (k/(\eps \delta))}}$ and it succeeds with
    probability $\geq 1-\delta$, where $\NSize = n  d$ is the
    input size.
\end{theorem}

\begin{proof}
    Let us first bound the price of the generated approximation. The
    proof is by induction. For $\cardin{P} \leq 20k$ the claim
    trivially holds.  Otherwise, by induction, we know that
    $\Price{Y}{\FlatB} \leq 2 \PrcOpt{Y}{k} \leq 2 \PrcOpt{P}{k}$.
    Also, $\Price{X}{\FlatA} \leq \ConstSL{} k^2 \cdot \PrcOpt{P}{k}$,
    by \lemref{r:sample}. As such, $\Price{P}{\Span(\FlatA \cup
       \FlatB)} \leq \Price{X}{\FlatA} + \Price{Y}{\FlatB} \leq 98 k^2
    \PrcOpt{P}{k}$. By \lemref{extract:k:flat}, the $k$-flat $\FlatC$
    is of price $\leq \ConstSAlg k^2 \PrcOpt{P}{k}$. Thus, by
    \lemref{easy}, we have $\Price{P}{\FlatD} \leq (1 + \eps)
    \PrcOpt{P}{k}$, as required.

    Next, we bound the probability of failure. We picked $\RSample$
    such that \lemref{r:sample} holds with probability $\geq
    1-(\delta/10)$. By induction, the recursive call succeeds with
    probability $\geq 1-\delta/2$. Finally, we used the algorithm of
    \lemref{easy} so that it succeeds with probability $\geq
    1-(\delta/10)$. Putting everything together, we have that the
    probability of failure is bounded by $\delta/10 + \delta/2 +
    \delta/10 \leq \delta$, as required.

    As for the running time, we have $T(20k, \delta) = O( d k + k^3)$.
    Next,
    \begin{eqnarray*}
        T(n, \delta) &=& T(n/4,\delta/2) + 
        O\pth{ \NSize k^2 \log (k/\delta)
           + n \pth{ k^2 \log (k/\delta)}^2 + \NSize (k/\eps) 
           \log(k/(\eps \delta)) }\\
        &=& O \pth{\NSize k  \pth{ \frac{1}{\eps} + k } 
           { \log \frac{k}{\eps \delta}}
          },
    \end{eqnarray*}
    assuming $k^2 \log(k/\delta) \leq d$.
\end{proof}

\begin{remark}    
    If $k^2 \log(k/\delta) \geq d$ then we just use the standard \SVD
    algorithm to extract the best $k$-flat in time $O(n d^2)$, which
    is better than the performance guaranteed by \thmref{main} in this
    case.
    
    Note, that for a fixed quality of approximation and constant
    confidence, the algorithm of \thmref{main} runs in $O(\NSize k^2
    \log k)$ time, and outputs a $k$-flat as required.

    Also, by putting all the random samples used by \thmref{main}
    together. We get a sample $\RSample$ of size $O( k^2
    \log(k/\delta) \log n + k^4/\eps )$ which spans a flat that
    contains a $k$-flat which is the required
    $(1+\eps)$-approximation.
\end{remark}

\begin{remark}
    If the input is given as a matrix $\MatA$ with $n$ rows and $d$
    columns, we transform it (conceptually) to a point set $P$ of $n$
    points in $\Re^d$.  The algorithm of \thmref{main} applied to $P$
    computes a $k$-flat $\FlatD$, which we can interpret as a $k$-rank
    matrix approximation to $\MatA$, by projecting the points of $P$
    onto $\FlatD$, and writing down the $i$\th projected point as the
    $i$\th row in the generated matrix $\MatB$.

    In fact, \thmref{main} provides an approximation to the optimal
    $k$-rank matrix under the Frobenius norm and not only the
    \emphi{squared} Frobenius norm. Specifically, given a matrix
    $\MatA$, \thmref{main} computes a matrix $\MatB$ of rank $k$, such
    that
    \begin{align*}
        \Fnorm{\MatA - \MatB}%
        = %
        \sqrt{ \Price{P}{\FlatD} }%
        \leq%
        \sqrt{(1+\eps)\PrcOpt{P}{k}} \leq (1+\eps) \min_{\MatC \text{
              matrix of rank } k} \Fnorm{\MatA - \MatC}.
    \end{align*}
\end{remark}

\section{Conclusions}

In this paper we presented a linear-time algorithm for low-rank matrix
approximation.  We believe that our techniques and more geometric
interpretation of the problem is of independent interest.  Note, that
our algorithm is not pass efficient since it requires $O( \log n)$
passes over the input (note however that the algorithm is IO
efficient). We leave the problem of how to develop a pass efficient
algorithm as an open problem for further research.

Surprisingly, the current bottleneck in our algorithm is the sampling
lemma (\lemref{r:sample}). It is natural to ask if it can be further
improved. In fact, the author is unaware of a way of proving this
lemma without using the ellipsoid argument (via $\eps$-net or random
sampling \cite{c-narsc-87} techniques), and an alternative proof
avoiding this might be interesting.  We leave this as an open problem
for further research.


\section*{Acknowledgments}

This result emerged during insightful discussions with Ke Chen. The
author would also like to thank Muthu Muthukrishnan and Piotr Indyk
for useful comments on the problem studied in this paper.

 
 \providecommand{\CNFX}[1]{ {\em{\textrm{(#1)}}}}
  \providecommand{\CNFSoCG}{\CNFX{SoCG}}
  \providecommand{\CNFCCCG}{\CNFX{CCCG}}
  \providecommand{\CNFFOCS}{\CNFX{FOCS}}
  \providecommand{\CNFSODA}{\CNFX{SODA}}
  \providecommand{\CNFSTOC}{\CNFX{STOC}}
  \providecommand{\CNFPODS}{\CNFX{PODS}}
  \providecommand{\CNFISAAC}{\CNFX{ISAAC}}
  \providecommand{\CNFFSTTCS}{\CNFX{FSTTCS}}
  \providecommand{\CNFIJCAI}{\CNFX{IJCAI}}
  \providecommand{\CNFBROADNETS}{\CNFX{BROADNETS}}
  \providecommand{\CNFLICS}{\CNFX{LICS}}
  \providecommand{\CNFSWAT}{\CNFX{SWAT}}  \providecommand{\CNFESA}{\CNFX{ESA}}
   \providecommand{\CNFX}[1]{{\em{\textrm{(#1)}}}}
  \providecommand{\CNFX}[1]{{\em{\textrm{(#1)}}}}
  \providecommand{\CNFCCCG}{\CNFX{CCCG}}
  \providecommand{\CNFFOCS}{\CNFX{FOCS}}
  \providecommand{\CNFSODA}{\CNFX{SODA}}
  \providecommand{\CNFSTOC}{\CNFX{STOC}}
  \providecommand{\CNFPODS}{\CNFX{PODS}}
  \providecommand{\CNFIPCO}{\CNFX{IPCO}}
  \providecommand{\CNFISAAC}{\CNFX{ISAAC}}
  \providecommand{\CNFFSTTCS}{\CNFX{FSTTCS}}
  \providecommand{\CNFIJCAI}{\CNFX{IJCAI}}
  \providecommand{\CNFBROADNETS}{\CNFX{BROADNETS}}
  \providecommand{\CNFESA}{\CNFX{ESA}}  \providecommand{\CNFSWAT}{\CNFX{SWAT}}
   \providecommand{\CNFWADS}{\CNFX{WADS}}
  \providecommand{\SarielWWWPapersAddr}{http://sarielhp.org/p}
  \providecommand{\SarielWWWPapers}{http://sarielhp.org/p}
  \providecommand{\urlSarielPaper}[1]{\href{\SarielWWWPapersAddr/#1}
  {\SarielWWWPapers{}/#1}}
  \providecommand{\tildegen}{{\protect\raisebox{-0.1cm}
  {\symbol{'176}\hspace{-0.03cm}}}}
  \providecommand{\CNFINFOCOM}{\CNFX{INFOCOM}}
  \providecommand{\Merigot}{M{\'{}e}rigot}
  \providecommand{\Badoiu}{B\u{a}doiu}  \providecommand{\Matousek}{Matou{\v
  s}ek}  \providecommand{\Erdos}{Erd{\H o}s}
  \providecommand{\Barany}{B{\'a}r{\'a}ny}
  \providecommand{\Bronimman}{Br{\"o}nnimann}
  \providecommand{\Gartner}{G{\"a}rtner}  \providecommand{\Badoiu}{B\u{a}doiu}


\end{document}